\documentclass[conference]{IEEEtran}

\usepackage{graphicx}
\usepackage{amsmath}
\usepackage{amssymb}
\usepackage{cite}      % Written by Donald Arseneau
\usepackage{multirow}

\usepackage{subfigure} % Written by Steven Douglas Cochran

\begin{document}

%\textwidth 6.5in
\title{Rateless Coding for MIMO Block Fading Channels}
%
%

\author{\authorblockN{Yijia Fan$^*$, Lifeng Lai$^*$, Elza Erkip$^{*\dagger}$, H. Vincent Poor$^*$}
\authorblockA{$^*$Department of Electrical Engineering, Princeton
University,
Princeton, NJ, 08544, USA\\
Email: \{yijiafan,llai,poor\}@princeton.edu \\
$^\dagger$Department of Electrical and Computer Engineering,
Polytechnic University, Brooklyn, NY, 11201, USA \\ Email:
elza@poly.edu}}

% instance, "A\textbf{} \textbf{}B" will typeset as "A B" not "AB". If you want
% "AB" then you have to do: "A\textbf{}\textbf{}B"
% *** Note that you probably will NOT want to include the author's name in ***
% *** the headers of peer review papers.                                   ***

% If you want to put a publisher's ID mark on the page
% (can leave text blank if you just want to see how the
% text height on the first page will be reduced by IEEE)
%\pubid{0000--0000/00\$00.00~\copyright~2002 IEEE}

%\specialpapernotice{(Invited Paper)}
\maketitle

\begin{abstract}
In this paper the performance limits and design principles of
rateless codes over fading channels are studied. The
diversity-multiplexing tradeoff (DMT) is used to analyze the system
performance for all possible transmission rates. It is revealed from
the analysis that the design of such rateless codes follows the
design principle of approximately universal codes for parallel
multiple-input multiple-output (MIMO) channels, in which each
sub-channel is a MIMO channel. More specifically, it is shown that
for a single-input single-output (SISO) channel, the previously
developed permutation codes of \emph{unit} length for parallel
channels having rate $LR$ can be transformed \emph{directly} into
rateless codes of length $L$ having multiple rate levels $(R, 2R,
\dots, LR)$, to achieve the DMT performance limit.
\end{abstract}


% Will be ignored for other modes.
%\IEEEpeerreviewmaketitle

\section{Introduction}

\newtheorem{definition}{Definition}
\newtheorem{lemma}{Lemma}
\newtheorem{theorem}{Theorem}
\newtheorem{corollary}{Corollary}
\newtheorem{remark}{Remark}
\newtheorem{numerical example}{Numerical example}

\subsection{Background}

Rateless codes present a class of codes that can be truncated to a
finite number of lengths, each of which has a certain likelihood of
being decoded to recover the entire message. Compared with
conventional coding schemes having a single rate $R$, such codes can
achieve multiple rate levels $(R, 2R, \dots, LR)$, depending on
different channel conditions. A rateless code is said to be
\emph{perfect} if each part of its codeword is capacity achieving.
Compared with conventional codes, rateless codes offer a potentially
\emph{higher rate}. Several results have been obtained on the design
of perfect rateless codes over erasure channels and additive white
Gaussian noise (AWGN) channels (see \cite{rt} and the references
therein).

Unlike in the fixed channel scenario, non-zero error probability
always exists in fading channels, when the instantaneous channel
state information (CSI) is not available at the transmitter and a
codeword spans only one or a small number of fading blocks. In this
scenario, it is well known that there is a fundamental tradeoff
between the information rate and error probability over fading
channels, which can be characterized as the diversity-multiplexing
tradeoff (DMT) \cite{DMT}.

\begin{definition}[DMT]
Consider a multiple-input multiple-output (MIMO) system and a family
of codes $C_\eta$ operating at average SNR $\eta$ per receive
antenna and having rates $R$. The multiplexing gain and diversity
order are defined as
\begin{equation}
r \buildrel \Delta \over = \mathop {\lim }\limits_{\eta  \to \infty
} \frac{{R}}{{\log _2 \eta }} \ \ \mathrm{and} \ \ {\rm{ }}d
\buildrel \Delta \over = - \mathop {\lim }\limits_{\eta  \to \infty
} \frac{{\log _2 P_e \left( R \right)}}{{\log _2 \eta }},
\end{equation}
where ${P_e}\left( R \right)$ is the average error probability at
the transmission rate $R$.
\end{definition}

The DMT is an effective performance measure for implementing the
rateless coding principles in a fading channel. Two main concerns
naturally arise: (a) determining the DMT limit for rateless coding
with finite numbers of blocks in a fading environment and
discovering how it performs with regard to conventional schemes; and
(b) determining DMT achieving codes that are simple (in the sense of
encoding and decoding complexity).

\subsection{Contributions of the Paper}

In this paper, we analyze the DMT performance of rateless codes. The
results show that, compared with conventional coding schemes having
multiplexing gain $r_n$, rateless codes having multiple rates $(r_n,
2r_n, \dots, Lr_n)$ offer an \emph{effective} multiplexing gain $r$
of $Lr_n$, given the same diversity gain at every rate, when $r_n$
is \emph{small}. As $r_n$ increases, the performance of rateless
codes degrades and ultimately becomes the same as that of
conventional schemes. Also while increasing $L$ lifts up the overall
system DMT curve, it does not necessarily improve the system
multiplexing gain for every fixed value of $r_n$. It is then
revealed that the design of such rateless codes follows the
principle of parallel channel codes that are \emph{approximately
universal} \cite{pcd} over fading channels. More specifically, it is
shown that for a single-input single-output (SISO) channel, the
formerly developed \emph{unit} length permutation codes for parallel
channels \cite{pcd} having rate $LR$ can be transformed
\emph{directly} into rateless codes of $L$-length having multiple
rate levels $(R, 2R, \dots, LR)$, to achieve the DMT performance
limit. For multiple-input multiple-output (MIMO) channels, the
results in the paper suggest a type of rateless codes that may be
viewed as a combination of conventional MIMO space-time codes and
parallel channel codes, both of which have been designed for fading
channels.

\subsection{Related Work}

The performance of rateless coding over fading channels has also
been considered in \cite{castura}, in which the throughput and error
probability are discussed. However, the tradeoff between these two
was not analyzed explicitly. For example, the results in
\cite{castura} shows that increasing the value of $L$ will decrease
the system error probability in certain scenario and is therefore
desirable. In this paper we show that while this discovery is true,
the system throughput, i.e., multiplexing gain might decrease when
$L$ becomes larger for every fixed value of $r_n$. Overall, our
results reveal that the optimal design of rateless codes requires
the consideration of both $r_n$ and $L$.

Rateless coding may be considered as a type of Hybrid-ARQ scheme
\cite{arq}. The DMT for ARQ has been revealed in \cite{arq}.
However, it will be shown in the paper that this DMT curve was
incomplete and represents the performance only when $r_n < \min
(M,N)/L$ in which $M$ and $N$ are the number of transmit and receive
antennas. The \emph{complete} DMT curve for rateless coding
including those parts for higher $r_n$ has never been revealed
before, and will be shown in this paper. In addition to this, the
results in this paper also offer a relationship between the design
parameter (i.e., $r_n$ and $L$) and the effective multiplexing gain
$r$ of the system, thus offer further insights into system design
and operational meaning compared to conventional coding schemes.
Furthermore, we suggest new design solutions for rateless codes.
Previous work on finite-rate feedback MIMO channels relies on either
power control or adaptive modulation and coding (e.g., \cite{dmtt}),
which are not necessary for our scheme.

The rest of this paper is organized as follows. The system model is
proposed in Section II. In Section III, the DMT performance of
rateless codes is studied. In Section IV, design of specific
rateless codes over fading channels is discussed. Finally,
concluding remarks are made in Section V.

\section{System Model}

We consider a frequency-flat fading channel with $M$ transmit
antennas and $N$ receive antennas. We assume that the transmitter
does not know the instantaneous CSI on its corresponding forward
channels, while CSI is available at the receiver. Each message is
encoded into a codeword of $L$ blocks. Each block takes $T$ channel
uses. We assume that the channel remains static for the entire
codeword length (i.e., $L$ blocks)\footnote{Note, however, that the
analysis in the paper can be extended straightforwardly to a faster
fading scenario in which the channel varies from block to block
during each codeword transmission.}. The system input-output
relationship can be expressed as
\begin{equation}
{\bf{Y}} = \sqrt {\frac{P} {M}} {\bf{HX}} + {\bf{N}} \label{ior}
\end{equation}
where ${\mathbf{X}} \in \mathbb{C}^{M \times TL}$ is the input
signal matrix; ${\mathbf{H}} \in \mathbb{C}^{N \times M}$ is the
channel transfer matrix whose elements are independent and
identically distributed (i.i.d.) complex Gaussian random variables
with zero means and unit variances; ${\bf{N}} \in \mathbb{C}^{N
\times TL}$ is the AWGN matrix with zero mean and covariance matrix
$\bf{I}$; and ${\bf{Y}} \in \mathbb{C}^{N \times TL}$ is the output
signal matrix. $P$ is the total transmit power, which also
corresponds to the average SNR $\eta$ (per receive antenna) at the
receiver.

The input signal matrix $\bf{X}$ can be written as
\begin{equation}
{\mathbf{X}} =
\left[ {\begin{array}{*{20}c}
   {{\mathbf{X}}_1 } &  \cdots  & {{\mathbf{X}}_L }  \\
 \end{array} } \right] \label{code}
\end{equation}
where ${\bf{X}}_l \in \mathbb{C}^{M \times T}$ is the codeword
matrix being sent during the $l$th block, and its corresponding
receiver noise matrix is denoted by ${\bf{N}}_l \in {\mathbb{C}}^{N
\times T}$. We impose a power constraint on each ${\bf{X}}_l$ so
that\footnote{Note that this is a more strict constraint than
letting $E\left[ {\frac{1} {TL}\left\| {{\mathbf{X}} } \right\|_F^2
} \right] \leqslant M$, which offers at least the same performance.}
\begin{equation}
E\left[ {\frac{1} {T}\left\| {{\mathbf{X}}_l } \right\|_F^2 }
\right] \leqslant M,
\end{equation}
for $l=1,...,L$.

\subsection{Conventional Schemes}

Assume that the transmitter sends the codeword at a rate $R$ bits
per channel use. A message of size $RT$ is encoded into a codeword
${\bf{X}}_l$ ($l=1,\dots,L$) and transmitted in $T$ channel uses. An
alternative method is to encode a message of size $RLT$ into
$\bf{X}$. Both encoding methods will offer the same performance
provided that $T$ is sufficiently large.

\subsection{Rateless Coding}

When rateless coding is applied, we wish to decode a message of size
$RLT$ with the codeword structure as shown in (\ref{code}). During
the transmission, the receiver measures the total mutual information
$I$ between the transmitter and the receiver and compares it with
$RLT$ after it receives each codeword block ${\bf{X}}_l$. If $I <
RLT$ after the $l$th block, the receiver remains silent and waits
for the next block. If $I \ge RLT$ after the $l$th block, it decodes
the received codeword $\left[ {\begin{array}{*{20}c}
   {{\mathbf{X}}_1 } &  \cdots  & {{\mathbf{X}}_l }  \\
 \end{array} } \right] $ and sends one bit of
positive feedback to the transmitter. Upon receiving the feedback,
the transmitter stops transmitting the remaining part of the current
codeword and starts transmitting the next message immediately.

Unlike conventional schemes, this process will bring multiple rate
levels $(R, 2R, \dots, LR)$. For example, if $I \ge RLT$ after the
first block is received (i.e., $l=1$) , the receiver will be able to
decode the entire message and the rate becomes $LR$. Similar
observations can be made for $l=2 \dots L$. Therefore, compared with
conventional schemes, the corresponding transmission rate achieved
by using rateless codes is always \emph{equal or higher}.
Specifically, we define the multiplexing gain for each rate level as
$(r_n, 2r_n, \dots, Lr_n)$ where \[r_n \buildrel \Delta \over =
\mathop {\lim }\limits_{\eta \to \infty } \frac{{R}}{{\log _2 \eta
}}.\] Later we will show through the DMT analysis that rateless
coding can retain the same diversity gain as conventional schemes,
but with a much higher multiplexing gain especially when the
corresponding $r_n$ is low.

\section{Performance analysis}

Denote by $\varepsilon _l$ the decoding error when decoding is
performed after the $l$th block ($0 \le l \le L$) and by $\Pr \left(
{\varepsilon _l, l } \right)$ the joint probability that a decoding
error occurs and decoding is achieved after $l$th block. The system
overall error probability can be expressed as
\[
P_e  = \sum\limits_{l = 1}^L {\Pr \left( {\varepsilon _l ,l}
\right)}.
\]

Define $p\left( l \right)$ ($0 \le l \le L$) to be the probability
with which $I < RLT$ after the $l$th block, and note that $p\left( 0
\right)=1$. Following the steps in Section II.B in \cite{arq}, the
average transmission rate for\emph{ each message }in bits per
channel use is given by
\begin{equation}
\bar R = \frac{{RL}} {{\sum\limits_{l = 0}^{L - 1} {p\left( l
\right)} }}. \label{trate}
\end{equation}
Note that this $\bar R$ describes the average rate with which the
message is removed from the \emph{transmitter}; i.e., it quantifies
how quickly the message is decoded at the receiver. We define the
effective multiplexing gain of the system as
\[
r = \mathop {\lim }\limits_{\eta  \to  + \infty } \frac{{\bar R}}
{{\log _2 \eta }}.
\]
Define $f\left( {k} \right)$ to be the piecewise linear function
connecting the points $\left(k, \left( {M - k} \right)\left( {N - k}
\right)\right)$ for integral $k=0,...,\min(M,N)$. Recall that a
conventional scheme operating at multiplexing gain $r_n$ ($0 \le r_n
\le \min (M,N)$) would have the diversity gain $f \left(r_n\right)$.
The following theorem shows the performance of rateless coding for
$0 \le r_n < +\infty$.

\begin{theorem}
Assume a sufficiently large $T$. For rateless codes having multiple
multiplexing gain levels $(r_n, 2r_n, \dots, Lr_n)$, the
corresponding DMT can be expressed as $(r,d)$ where \[r = r_n \cdot
\frac{L} {l} \ \ {\rm{and}} \ \ d =f\left( {\frac{lr} {L}} \right)\]
for
\[ \frac{{l - 1}} {L}\min \left(
{M,N} \right) \leqslant r_n  < \frac{l} {L}\min \left( {M,N} \right)
\]
and $l=1,2,...L$. Finally, $d=0$ for $r_n \ge \min(M,N)$.
\end{theorem}
\begin{proof}
See Appendix A.
\end{proof}

Note that for rateless coding to achieve the performance in
\emph{Theorem 1}, we do not necessarily require $T \rightarrow
+\infty$. As long as $T$ is large enough such that the error
probability $\Pr \left( {\varepsilon _l, l} \right) \mathop
\leqslant \limits^. \eta^{f\left( {r_n } \right)}$ for each $l$, the
DMT in \emph{Theorem 1} can be achieved. While the minimal $T$ for a
general MIMO channel when applying rateless coding is unknown to the
authors, it will be shown later that for SISO channels, $T=1$ is
sufficient to achieve the optimal DMT in \emph{Theorem 1}.

Comparing rateless coding with conventional schemes, it can be shown
that for $0 \le r_n <\min (M,N)/L$, $r=Lr_n$ for $d=f \left(r_n
\right)$. In this scenario rateless coding can improve the
multiplexing gain up to $L$ times that of conventional schemes,
given the same diversity gain. Fig. 1 gives an example when $M=N=2$
and $L=2$, and $0 \le r_n \le 1$. The operating point A in the curve
for a conventional scheme for $0 \le r_n \le 1$ corresponds to point
B in the curve for rateless coding.

\begin{figure}[t!]
\centering
\includegraphics[width=3.7in]{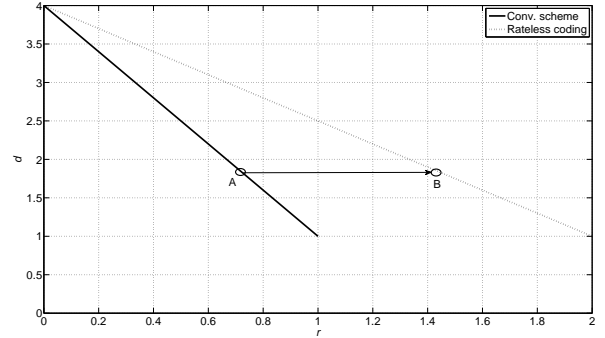}
% where an .eps filename suffix will be assumed under latex,
% and a .pdf suffix will be assumed for pdflatex
\caption{The DMTs for conventional schemes and rateless coding for
$0 \le r_n \le 1$. $M=N=2$, $L=2$.} \label{abq1}
\end{figure}

An important observation from \emph{Theorem 1} is that the system
performance will not be improved after $r_n$ (\emph{almost}) reaches
$\min(M,N)/L$, as the optimal DMT is already achieved by using
rateless coding. This is mainly due to the fact that the first block
can no longer support the message size when the message rate reaches
$\min(M,N)/L$. Thus the system multiplexing gain \emph{decreases}
for the same diversity gain, and finally offers the same DMT as
conventional schemes when the first $L-1$ blocks all fail to decode
the message. Fig. 2 shows an example when $M=N=3$, $L=4$. This
observation also implies that for any fixed value of $r_n$, simply
increasing the value of $L$ does \emph{not} \emph{necessarily}
improve the system DMT performance. Although the overall system DMT
will increase when $L$ is larger, the multiplexing gain might
decrease for certain fixed values of $r_n$. A convenient choice for
$L$ would be in the region of $L < \min(M,N)/r_n$. However, note
that the maximal multiplexing gain $\min (M,N)$ can be achieved only
with zero diversity gain, and this happens when $r_n = \min (M,N)$
\emph{regardless of} the value of $L$.

\begin{figure}[t!]
\centering
\includegraphics[width=3.7in]{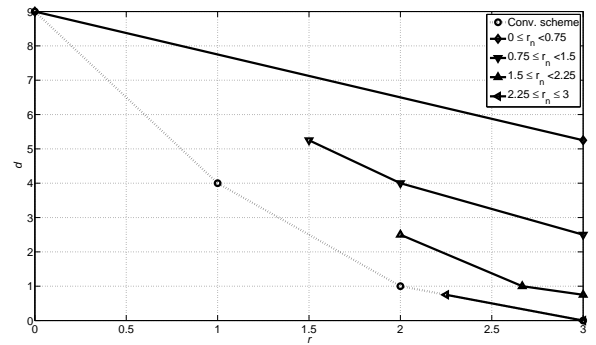}
% where an .eps filename suffix will be assumed under latex,
% and a .pdf suffix will be assumed for pdflatex
\caption{The DMTs for different schemes for $0 \le r_n \le 3$.
$M=N=3$, $L=4$.} \label{dmtfb}
\end{figure}

\section{Design of rateless codes}

Note that codewords ${\bf{X}}_i$ ($1 \le i \le L$) in (\ref{code})
are transmitted through different channels that are
\emph{orthogonal} in time. This is analogous to transmitting
${\bf{X}}_i$ through different channels that are parallel in
\emph{space}. In the (space) parallel channel model, elements in
$\left\{{\bf{X}}_i \right\}$ can be jointly (simultaneously)
decoded. However, for the channel model considered in this paper,
which we now call the \emph{rateless channel}, the decoding process
needs to follow certain direction in time, i.e., we start decoding
from ${\bf{X}}_1$, then $\left[{\bf{X}}_1 \ \ {\bf{X}}_2\right]$ if
${\bf{X}}_1$ is not decoded, etc. This comparison implies that while
good parallel channel codes can be used as the basis for rateless
coding, they might need modifications in order to offer good
performance over the rateless channel.

Specifically, for the rateless channel expressed in the form of
(\ref{ior}), we consider the corresponding parallel MIMO channel, in
which each sub-channel is a MIMO channel, having the following
input-output relationship:
\begin{eqnarray}
{\mathbf{Y}} = \sqrt {\frac{P}{M}} \left( {\begin{array}{*{20}c}
   {{\mathbf{H}} } & {} & \text{\Large{0}}  \\
   {} &  \ddots  & {}  \\
   \text{\Large{0}} & {} & {{\mathbf{H}} }  \\

 \end{array} } \right)\left( \begin{gathered}
  {\mathbf{X}}_1  \hfill \\
   \vdots  \hfill \\
  {\mathbf{X}}_L  \hfill \\
\end{gathered}  \right) + \left( \begin{gathered}
  {\mathbf{N}}_1  \hfill \\
   \vdots  \hfill \\
  {\mathbf{N}}_L  \hfill \\
\end{gathered}  \right)
\label{iop}
\end{eqnarray}
where $\bf{H}$, ${\bf{X}}_i$ and ${\bf{N}}_i$ are the same as those
in (\ref{ior}). It is easy to see that the DMT for this system is $d
= f\left( {\frac{r} {L}} \right)$ for $0 \le r \le L\min(M,N)$.
Assuming a code that achieves this DMT, when we implement its
transformation $\left[ {\begin{array}{*{20}c}
   {{\mathbf{X}}_1 } &  \cdots  & {{\mathbf{X}}_L }  \\
 \end{array} } \right]$ into the rateless channel having multiple rates $(r_n, 2r_n, \dots, Lr_n)$, it is
not difficult to show
 that
 \begin{equation}
\Pr \left( {\varepsilon _L ,L} \right) \mathop  \leqslant \limits^.
\eta ^{ -f\left( {r_n} \right)}. \label{parlc}
\end{equation}
In order to make the overall $P_e \mathop \leqslant \limits^. \eta
^{ -f\left( r_n \right)}$, we need to ensure that $\Pr \left(
{\varepsilon _l ,l} \right) \mathop \leqslant \limits^. \eta ^{
-f\left( r_n \right)}$ for $1 \le l \le L-1$. However, those
conditions are not \emph{essential} in order to achieve the optimal
DMT for the parallel channel shown in (\ref{iop}), which only
requires the condition (\ref{parlc}). Thus stricter code design
criteria are required for the rateless channel. One example of such
a criterion is the \emph{approximately universal} criterion
\cite{pcd}.

Codes being \emph{approximately universal} for parallel channels
ensure that the highest error probability when decoding \emph{any}
subset of $\{{\bf{X}}_i\}$ in the set of all non-outage events
decays \emph{exponentially} in SNR (i.e., in the form of
$e^{-\eta^\delta}$ for some $\delta >0$) under \emph{any }fading
distribution, and thus can be ignored compared with the outage
probability under the same fading distribution, when the SNR goes to
infinity. Specifically, we consider the following parallel MIMO
channel which is more general than the one in (\ref{iop}):
\begin{eqnarray}
{\mathbf{Y}} = \sqrt {\frac{P}{M}} \left( {\begin{array}{*{20}c}
   {{\mathbf{H}}_1 } & {} & \text{\Large{0}}  \\
   {} &  \ddots  & {}  \\
   \text{\Large{0}} & {} & {{\mathbf{H}}_L }  \\

 \end{array} } \right)\left( \begin{gathered}
  {\mathbf{X}}_1  \hfill \\
   \vdots  \hfill \\
  {\mathbf{X}}_L  \hfill \\
\end{gathered}  \right) + \left( \begin{gathered}
  {\mathbf{N}}_1  \hfill \\
   \vdots  \hfill \\
  {\mathbf{N}}_L  \hfill \\
\end{gathered}  \right)
\label{iop2}
\end{eqnarray}
where each channel matrix in $\{{\bf{H}}_i\}$ ($1 \le i \le L$)
follows an \emph{arbitrary} distribution. In particular, when the
matrices in $\{{\bf{H}}_i\}$ are i.i.d. and of the same
distributions as the $\bf{H}$ in (\ref{ior}), following the same
steps as those in \cite{DMT}, it is not difficult to show that the
optimal DMT for this system is $d = L f\left( {\frac{r} {L}}
\right)$ for $0 \le r \le L\min(M,N)$. Now, we are ready to state
the following theorem considering the performance of rateless codes
that are transformed from the approximately universal codes for the
parallel channel in (\ref{iop2}).
\begin{theorem}
Suppose a code $\left[ {\begin{array}{*{20}c}
   {{\mathbf{X}}_1^T } &  \cdots  & {{\mathbf{X}}_L^T }  \\
 \end{array} } \right]^T$ is \emph{approximately universal} for the parallel
 channel shown in (\ref{iop2}) and can achieve the DMT points
 $(Lr_n, Lf\left( {r_n} \right))$ for $0 \le r_n \le \min
 (M,N)$ when the channel matrices have i.i.d. Rayleigh fading. Then, its transformation $\left[ {\begin{array}{*{20}c}
   {{\mathbf{X}}_1 } &  \cdots  & {{\mathbf{X}}_L }  \\
 \end{array} } \right]$, when applied to the rateless channel shown in (\ref{ior}) aiming at multiple multiplexing gains $(r_n, 2r_n, \dots, Lr_n)$,
 can achieve the DMT shown in \emph{Theorem
 1}.
\end{theorem}
\begin{proof}
See Appendix B.
\end{proof}
While approximately universal codes for the general parallel MIMO
channel is unknown to the authors, approximately universal codes for
parallel SISO channels do exist, and can be transformed directly
into good rateless codes for SISO channels. In the following, we
apply permutation codes for parallel channels \cite{pcd} to the
rateless channel.

Permutation codes are a class of codes generated from QAM
constellations. In the encoding process, a message is mapped into
different QAM constellation points across all subchannels. The
constellation over one subchannel is a permutation of the points in
the constellation over any other subchannel. The permutation is
optimized such that the minimal codeword difference is large enough
to satisfy the approximate universality criterion. Explicit
permutation codes can be constructed using \emph{universally
decodable matrices}. We refer the readers to \cite{pcd} and the
references therein for details. It has been shown that permutation
codes achieve the optimal DMT for parallel channels and have a
particularly simple structure. For example, the codewords are of
\emph{unit} length.

Assume the transmission rates over rateless channel are $(R, 2R,
\dots, LR)$ bits per channel use. To implement permutation codes, we
choose a codebook of size $2^{LR}$ (messages) for the parallel
channel in (\ref{iop2}). Each message is mapped into a code $\left[
{\begin{array}{*{20}c}
   {{\mathbf{X}}_1^T } &  \cdots  & {{\mathbf{X}}_L^T }  \\
 \end{array} } \right]^T$, in which each
 ${\mathbf{X}}_l$ is an $2^{LR}$-point QAM constellation. The message
 can be fully recovered as long as any subset of $\left\{{\mathbf{X}}_l\right\}$ can be
 correctly decoded. Now, we transform this code into the form $\left[ {\begin{array}{*{20}c}
   {{\mathbf{X}}_1 } &  \cdots  & {{\mathbf{X}}_L }  \\
 \end{array} } \right]$ for the rateless channel. Since $\Pr \left( {\varepsilon _l ,l} \right)$ decays exponentially
in SNR due to the approximate universality of such codes, the
overall error probability is always dominated by that upon receiving
all ${\bf{X}}_l$ for \emph{infinitely} high SNR. More precisely, we
summarize the above observations as the following corollary.

\begin{corollary}
Rateless codes that are transformed from permutation codes for
parallel channels can offer exactly the same performance as shown in
\emph{Theorem 1} over the SISO rateless channel.
\end{corollary}
\begin{proof}
The proof is a direct extension of the proof of \emph{Theorem 2} and
is omitted.
\end{proof}

\section{Conclusions}

The performance of rateless codes has been studied for MIMO fading
channels in terms of the DMT. The analysis shows that design
principles for rateless codes can follow these of the approximately
universal codes for parallel MIMO channels. Specifically, it has
been shown that for a SISO channel, the formerly developed
permutation codes of \emph{unit} length for parallel channels having
rate $LR$ can be transformed \emph{directly} into rateless codes of
length $L$ having multiple rate levels $(R, 2R, \dots, LR)$, to
achieve the desired optimal DMT performance.

\appendix

\subsection{Proof of Theorem 1}
Define $r_L=Lr_n$. Following the steps in \cite{DMT}, it is easy to
show that $p \left(l \right)\doteq \eta ^{-f\left( {\frac{r_L} {l}}
\right)}$ for $l \neq 0$. We write the error probability as
\begin{equation}
P_e  = \sum\limits_{l = 1}^{L - 1} {(1 - p \left( l \right))\Pr
\left( {\varepsilon _l } \right)}  + \Pr \left( {\varepsilon _L, L}
\right). \label{ub}
\end{equation}
In (\ref{ub}), $\Pr \left( {\varepsilon _l } \right)$ is error
probability when $lI_b \ge LTR$, where $I_b$ is the mutual
information of the channel in each block. Using Fano's inequality we
can obtain the error probability lower bound \cite{DMT}:
\[P_e \ge \Pr \left( {\varepsilon _L, L} \right) \mathop  \geqslant \limits^.
\eta ^{ -f\left( {\frac{r_L} {L}} \right)}. \] Since $r \le r_L$, we
have $\eta ^{ -f\left( {\frac{r_L} {L}} \right)}  \ge \eta ^{ -
f\left( {\frac{r} {L}} \right)}$, and thus the desired performance
upper bound is obtained.

Now we prove the achievability part. Consider ${\Pr \left(
{\varepsilon _l } \right)}$. Following the same argument as in the
proof of Theorem 10.1.1 in \cite{14}, we get
\begin{equation}
\Pr \left( {\varepsilon _l } \right) \leqslant 3\epsilon \label{gau}
\end{equation}
for sufficiently large $T$. Note that a very similar argument has
been made in \emph{Lemma 1} in \cite{7}, although it is claimed
there that both $T$ and $L$ are required to be sufficiently large in
order to satisfy (\ref{gau}). Now $(\ref{ub})$ can be further
rewritten as
\begin{eqnarray}
P_e &\leqslant& 3(L-1) \epsilon + \eta ^{ -f\left( {\frac{r_L} {L}}
\right)} + (1-p \left( L \right))\Pr \left( {\varepsilon _L} \right)
\nonumber
\\ &\doteq& \eta ^{ -f\left( {\frac{r_L} {L}}
\right)}.
\end{eqnarray}
Note that
\[
\bar R \doteq \frac{{LR}} {{1 + \sum\limits_{i = 1}^{L - 1} {\eta ^{
- f\left( {\frac{r_L} {l}} \right)} } }} \doteq LR
\]
for $0 \le r_L < \min (M,N)$. Thus $r=r_L$ and diversity gain
$f\left( {\frac{r} {L}} \right)$ is achievable in the range $0 \le r
< \min (M,N)$. Note that $r_L=Lr_n$, and thus we have $d=f\left( r_n
\right)$ for
\[r = r_n L, 0 \le r_n < \frac{\min(M,N)}{L}.
\]

So far we have only considered the scenario in which
$r_n<\frac{\min(M,N)}{L}$. Now the question to ask is what happens
if we increase the value of $r_n$ to $\frac{\min(M,N)}{L}$ and
beyond. In this scenario, $f\left( {\frac{r_L} {1}} \right) = 0$,
and thus $\bar R \doteq \frac{{LR}} {2}.$ The message rate $r$ is
decreased to $r_L/2$ due to the fact that after the first block the
receiver has no chance of decoding the message correctly and it
always needs the second block. However, the system error probability
$P_e$ is not changed. Therefore the message rate becomes
\begin{equation}
r=r_n \cdot\frac{L}{2}, \frac{\min(M,N)}{L} \le r_n <
\frac{2\min(M,N)}{L},
\end{equation}
and the system DMT becomes
\begin{equation}
d = f\left( {\frac{2r} {L}} \right) , \frac{\min(M,N)}{2} \le r <
\min(M,N).
\end{equation}
Similarly, when $r$ reaches $\min(M,N)$ again, i.e., $r_n$ reaches
$\frac{2\min(M,N)}{L}$, $f\left( {\frac{r_L} {2}} \right)= f\left(
{\frac{2r} {2}} \right) = 0. $ Thus $\bar R \doteq \frac{{LR}} {3}$
and
\begin{equation}
r=r_n \cdot \frac{L}{3}, \frac{2\min(M,N)}{L} \le r_n <
\frac{3\min(M,N)}{L};
\end{equation}
the system DMT becomes
\begin{equation}
d = f\left( {\frac{3r} {L}} \right) , \frac{2\min(M,N)}{3} \le r <
\min(M,N).
\end{equation}
Continuing following the above until $ \bar R \doteq R$, we obtain
the desired result and the proof is completed.

\subsection{Proof of Theorem 2}
Assume that the system in (\ref{iop}) transmits at a rate $LR=r_L
\log _2 \eta$. The probability of any decoding error can be upper
bounded by \cite{DMT}
\[
P  \leqslant P_O  + P_{e|O^c }
\]
where $P_O$ is the outage probability and $P_{e|O^c }$ is the
average error probability given that the channel is not in outage.
Approximately universality means that for such codes $P_{e|O^c }
=e^{-\eta^\delta}$ under \emph{any} fading distribution. For the
system in (\ref{iop2}), these include the fading distributions in
which ${\bf{H}}_1= \dots = {\bf{H}}_l$ follow the same distribution
as the $\bf{H}$ in (\ref{ior}) and ${\bf{H}}_{l+1}=\dots={\bf{H}}_L
\equiv 0$ for all $1 \le l \le L-1$. When such codes are transformed
into the rateless channels shown in (\ref{ior}), it is a simple
matter to show that
\[
\Pr \left( {\varepsilon _l } \right) = P_{e|O^c } =e^{-\eta^\delta}
\]
for any $1 \le l \le L$, where $\Pr \left( {\varepsilon _l }
\right)$ is given in (\ref{ub}). Thus the system error probability
for the rateless channel in (\ref{ior}) is always upper bounded by
\begin{eqnarray}
P_e  \leqslant Le^{ - \eta ^\delta  }  + \eta ^{ - f\left(
{\frac{{r_L }} {L}} \right)} \nonumber \doteq \eta ^{ - f\left(
{\frac{{r_L }} {L}} \right)}.
\end{eqnarray}
The rest of the proof follows that of \emph{Theorem 1} and is
omitted.

% or

% Use this command to get the appendices' numbers in "A", "B" instead of the
% default capitalized Roman numerals ("I", "II", etc.).
% However, the capital letter form may result in awkward subsection numbers
% (such as "A-A"). Capitalized Roman numerals are the default.
%\useRomanappendicesfalse
%


\section*{Acknowledgement}
This research was supported by the U.S. National Science Foundation
under Grants ANI-03-38807 and CNS-06-25637.


% optional entry into table of contents (if used)
%\addcontentsline{toc}{section}{Acknowledgment}


% http://www.ctan.org/tex-archive/macros/latex/contrib/supported/IEEEtran/bibtex
%\bibliographystyle{IEEEtran.bst}

\begin{thebibliography}{1}

\bibitem{DMT}
L. Zheng and D. Tse, ``Diversity and multiplexing: A fundamental
tradeoff in multiple antenna channels,'' \emph{IEEE Trans. Inf.
Theory}, vol. 49, no. 5, pp. 1073-1096, May 2003.

\bibitem{arq}
H. El Gamal, G. Caire, M. O. Damen, ``The MIMO ARQ channel:
Diversity-multiplexing-delay tradeoff,'' \emph{IEEE. Trans. Inf.
Theory.}, vol. 52, no. 8, pp. 3601-3619, Aug. 2006.

\bibitem{pcd}
S. Tavildar and P. Viswanath, ``Approximately universal codes over
slow fading channels,'' \emph{IEEE Trans. Inf. Theory}, vol. 52, no.
7, pp. 3233-3258, Jul. 2006.

\bibitem{castura}
J. Castura, Y. Mao and S. Draper, ``On rateless coding over fading
channels with delay constraints,'' 2006 Int'l Sym. Inf. Theory (ISIT
2006), Seattle, USA, Jul., 2006.

\bibitem{dmtt}
T. T. Kim and M. Skoglund, ``Diversity-multiplexing tradeoff in MIMO
channels with partial CSIT,'' \emph{IEEE Trans. Inf. Theory}, vol.
53, no.8, pp. 2743-2759, Aug. 2007.

\bibitem{rt}
U. Erez, M. Trott and G. Wornell, ``Rateless Coding for Gaussian
Channels,'' submitted to \emph{IEEE Trans. Inf. Theory}, available
on \verb+arxiv.org/PS_cache/arxiv/pdf/0708/0708.2575v1.pdf+

\bibitem{7}
K. Azarian, H. El Gamal, and P. Schniter, ``On the achievable
diversity-multiplexing tradeoff in half-duplex cooperative
channels,'' \emph{IEEE Trans. Inf. Theory}, vol 51, no. 12 pp.
4152-4172, Dec. 2005.

\bibitem{14}
T. Cover and J. A. Thomas, \emph{Elements of Information Theory},
Wiley: New York, 1991.

\end{thebibliography}

%\renewcommand{\baselinestretch}{1.1}

\end{document}